\documentclass[]{amsart}

\usepackage{verbatim}
\usepackage{amssymb}
\usepackage{amsbsy}
\usepackage{amscd}
\usepackage{amsmath}
\usepackage{amsthm}
\usepackage{amsrefs}
\usepackage{graphicx}
\usepackage[mathscr]{eucal}
\usepackage[top=20mm,bottom=20mm,right=20mm,left=20mm]{geometry}

\newtheorem{theorem}{Theorem}
\newtheorem{lemma}{Lemma}
\theoremstyle{definition}
\newtheorem{defn}{Definition}\numberwithin{defn}{section}

\newcommand{\tr}{\operatorname{tr}}

\newcommand{\R}{\mathbb{R}}

\renewcommand{\H}{\mathbb{H}}
\newcommand{\C}{\mathbb{C}}

\newcommand{\dist}{\operatorname{dist}}

\title[Harmonic maps and general relativity]{Harmonic maps with prescribed singularities\\ and applications in general relativity}

\author[]{Gilbert Weinstein}

\address{Department of Mathematics \& Department of Physics\\
	Ariel University, Ariel, Israel 40700}

\thanks{The author wishes to express his thanks to the organizers of the MSJ-SI for their invitation to participate in this conference, and for their hospitality in Sapporo.}

\email{gilbertw@ariel.ac.il}

\subjclass[2010]{}


\begin{document}

\begin{abstract}
This paper presents a general existence and uniqueness result for harmonic maps with prescribed singularities into non-positively curved targets, and surveys a number of applications to general relativity. It is based on a talk delivered by the author at \emph{The 11th Mathematical Society of Japan  Seasonal Institute, The Role of Metrics in the Theory of Partial Differential Equations}.
\end{abstract}

\maketitle

\section{Introduction}

We begin with a brief review of harmonic maps. Let $(M,g)$ and $(N,h)$ be Riemannian manifolds and let $\varphi\colon M\to N$ be a smooth map. Then $d\varphi\colon M \to TM \otimes \varphi^{-1} TN$ is a cross-section of the tangent bundle of $M$ tensor-product with the pull back bundle of $TN$ under $\varphi$. In local coordinates we have:
\[
d\varphi = \frac{\partial\varphi^\mu}{\partial x^i} \, dx^i\otimes\frac{\partial}{\partial y^\mu}.
\]
There is a canonical positive definite inner product on this bundle with the norm expressed in local coordinates:
\[
|d\varphi|^2 = \frac{\partial \varphi^\mu}{\partial x^i}\, \frac{\partial \varphi^\nu}{\partial x^j}\, g^{ij}\, h_{\mu\nu}.
\]
From this we can form an energy:
\[
E_\Omega(\varphi) = \int_\Omega |d\varphi|^2\, dV_g
\]

\begin{defn}
	A map $\varphi\colon M\to N$ is \emph{harmonic} if for each $\Omega\subset\subset M$ it is a critical point of $E_\Omega$. The map $\varphi$ is an \emph{local energy minimizer} if for every $\Omega\subset\subset M$, $E_\Omega(\varphi) \leq E_\Omega(\psi)$ for all $\psi$ which agrees with $\varphi$ on the boundary of $\Omega$. It is an energy minimizer if the same is true with $\Omega=M$.
\end{defn}

Clearly every local energy minimizer is a harmonic map. The Euler-Lagrange equations read
\begin{equation} \label{harmmap}
\tau^\mu(\varphi)=\Delta_g \varphi^\mu + \Gamma^\mu_{\alpha\beta}(\varphi) \, g(\nabla\varphi^\alpha,\nabla\varphi^\beta) = 0,
\end{equation}
where $\Gamma^\mu_{\alpha\beta}(\varphi)$ are the Christoffel symbols of $N$ evaluated along $\varphi$, and $\Delta_g=\tr\nabla^2$ is the Laplace-Beltrami operator on $M$. This forms a system of semi-linear elliptic PDEs on $M$ with the nonlinearity being quadratic in the first derivatives of $\varphi$. For an arbitrary map $\varphi$, the left-hand side $\tau$ is a cross section of $\varphi^{-1} TN$ called the \emph{tension} of $\varphi$.

Familiar examples of harmonic maps are the two cases: (i) $\operatorname{dim} N=1$, then $\Gamma=0$ and the harmonic map problem reduces to the linear Laplace equation; (ii) $\operatorname{dim} M=1$, then the harmonic map problem becomes the ODE problem for geodesics. The first significant nonlinear PDE result concerning harmonic maps was obtained by Eells-Sampson~\cite{eellssampson}: \emph{If the sectional curvature of $N$ satisfies $\kappa_N\leq 0$, and $\varphi_0\colon M\to N$ is any smooth map, then there is a harmonic map $\varphi$ homotopic to $\varphi_0$.} Without the assumption on $\kappa_N$, harmonic maps and even energy minimizer can develop singularities, and the best one can do is control the size of the singular set~\cite{schoenuhlenbeck}.

\section{Harmonic maps with prescribed singularities} \label{harmap}

Consider first the linear case. Let $(M,g)$ be a complete Riemannian manifold\footnote{A similar result can be stated for a compact manifold with boundary.} and let $\Gamma\subset M$ be a submanifold of co-dimension $\geq 2$. Let $\rho\colon\Gamma\to\R^+$ be a positive continuous function, and let $G(x,y)$ be the Green function for $\Delta_g$ on $M$, then
\[
u(x) = \int_\Sigma \rho(y)\, G(x,y)\, dv_\Sigma(y).
\]
is a harmonic function on $M\setminus\Gamma$ which blows up on $\Gamma$, $\Delta_g u = \rho\delta_\Gamma$. Now if $\gamma$ is a geodesic in $N$, then $\gamma\circ u$ is a harmonic map into $N$ which blows up on $\Gamma$. Before we state and prove a nonlinear analog of this result, we need two definitions.

\begin{defn} \label{model}
	A \emph{model map} is a map $\varphi_0\colon M\setminus\Gamma\to N$ such that (i) $|\tau(\varphi_0)|$ is bounded; and (ii) $|\tau(\varphi_0)|$ decays "fast enough", in the sense that there is a function $v>0$ with $v\to0$ at infinity and $|\tau(\varphi_0)|\leq -\Delta_g v$.
\end{defn}

It turns out for instance that for $M=\R^3$, the case of interest for our applications, it is sufficient to require $|\tau(\varphi_0)|=O(r^{-2-\epsilon})$ with $\epsilon>0$.

\begin{defn}
	Two maps $\varphi_1,\varphi_2\colon M\setminus\Gamma\to N$ are \emph{asymptotic} if $\dist_N(\varphi_1,\varphi_2)$ is (i) bounded; and (ii) falls off to $0$ at infinity.
\end{defn}

Part (i) of this definition is inspired by the analogous definition for geodesics~\cite{eberleinoneill}.

\begin{theorem}	\label{existence-uniqueness}
	Let $\varphi_0\colon M\setminus\Gamma\to N$ be a model map, and let either $-b^2<\kappa_N<-a^2$, or $N$ be a symmetric space of non-compact type and rank $\leq 2$, then there is a unique harmonic map $\varphi\colon M\setminus\Gamma\to N$ which is asymptotic to $\varphi_0$.
\end{theorem}

\begin{proof}
For completeness, we give some of the details of the proof. We need the following Lemma~\cite{weinsteinMRL}.

\begin{lemma} \label{supbound}
Let $\varphi_1,\varphi_2\colon M\setminus\Gamma\to N$ be smooth maps, and suppose $\kappa_N\leq0$, then 
\[
	\Delta_g\sqrt{1+\dist_N(\varphi_1,\varphi_2)^2} \geq -(|\tau(\varphi_1)|+|\tau(\varphi_2)|).
\]
\end{lemma}

The Hessian $H$ of the distance function $d_N\colon N\times N\to\R$ is positive semi-definite, thanks to the hypothesis $\kappa_N\leq 0$~\cite{schoenyau1979compact}, hence denoting $d_p(\cdot)=d_N(p,\cdot)$,  we have
\[
	\Delta_g d_N(\varphi_1,\varphi_2) = \nabla d_{\varphi_1} \cdot \tau(\varphi_2)
	+ \tau(\varphi_1) \cdot \nabla d_{\varphi_2} + H(d\varphi_1+d\varphi_2,d\varphi_1+d\varphi_2)
	\geq -(|\tau(\varphi_1)|+|\tau(\varphi_2)|),
\]
wherever $d_N(\varphi_1,\varphi_2)>0$. Now, we get
\[
	\Delta\sqrt{1+d_N^2}\geq
	\frac{d_N}{\sqrt{1+d_N^2}}\, \Delta d_N \geq -\frac{d_N}{\sqrt{1+d_N^2}}\,(|\tau(\varphi_1)|+|\tau(\varphi_2)|)
	\geq -(|\tau(\varphi_1)|+|\tau(\varphi_2)|),
\]
when $d_N>0$, and since $\sqrt{1+d_N^2}\geq0$ has a global minimum on the set $\{d_N=0\}$, it follows trivially that $\Delta\sqrt{1+d_N^2}\geq0$ also there. This completes the proof of the Lemma.

We now return to the proof of the Theorem. We first establish uniqueness. If $\varphi_1,\varphi_2\colon M\setminus\Gamma\to N$ are both harmonic and asymptotic to $\varphi_0$, then $\varphi_1$ is asymptotic to $\varphi_2$ and by Lemma~\ref{supbound}, $u=d_N(\varphi_1,\varphi_2)$ is a bounded subharmonic function on $M\setminus\Gamma$ which tends to $0$ at infinity. Since $\Gamma$ is of co-dimension $\geq2$, it follows from Lemma~8 in~\cite{weinsteinDMJ} that $u$ is weakly harmonic on $M$. Hence by the maximum principle $u=0$ and $\varphi_1=\varphi_2$.

We now prove existence. Fix some point $x_0\in M\setminus\Gamma$, pick $\delta>0$, and let $M_\delta=B_{1/\delta}(x_0) \cap \{x\in M\colon d_M(x,\Gamma)>\delta\}$. Since $\overline M_\delta$ is compact and $\varphi_0$ is smooth on $\overline M_\delta$, there exists a harmonic map $\varphi_\delta\colon M_\delta\to N$ such that $\varphi_\delta=\varphi_0$ on $\partial M_\delta$. Let $u_\delta = \sqrt{1+d_N(\varphi_0,\varphi_\delta)^2}$, then by Lemma~\ref{supbound} and Definition~\ref{model}, we have
\[
	\Delta_g (u_\delta-v) \geq 0,
\]
on $M_\delta$, where $v$ is the function in Definition~\ref{model}, hence $v>0$ and $v\to0$ at infinity. Also $u_\delta-v\leq 1$ on $\partial M_\delta$, thus by the maximum principle it follows immediately that $u_\delta\leq C$ independent of $\delta>0$, and $d_N(\varphi_0,\varphi_\delta)\to0$ at infinity.

We now consider the maps $\varphi_\delta$ restricted to a fixed set $M_{\delta_0}$. These map into a fixed ball $B_R\subset N$, and hence we can obtain local pointwise a priori bounds uniform in $\delta$  not just for $\varphi_\delta$ but also for any number of derivatives. From here it is straightforward to obtain a sequence $\delta_i\to$ such that $\varphi_{\delta_i}$ converges uniformly on compact subsets of $M\setminus\Gamma$ together with two derivatives, and hence the limit $\varphi$ will be harmonic and asymptotic to $\varphi_0$.

To achieve this, the first step is to obtain a priori local energy bounds. These follow from the pointwise bound above and require a particular coordinate system on the target $N$ of the harmonic map:
\[
	ds^2 = \sum_{j=1}^k dr_j^2 + Q(\xi,\xi)
\]
where $1\leq k\leq 2$ is either 1 or the rank of $N$, and $Q$ is a quadratic form on $\R^{l-k}$ satisfying
\[
	2a\, Q(\xi,\xi) \leq \frac{d}{dr_j} Q(\xi,\xi) \leq 2b\, Q(\xi,\xi), \quad \forall \xi\in\R^{l-k},
\]
for some constants $0<a<b$. Here where $l$ is the dimension of $N$.
The existence of such a coordinate system for symmetric spaces of non-compact type and rank $2$ was proved in~\cite{khuriweinsteinyamada2017}. However, we conjecture the same holds true for any rank.
In this coordinate system, $k$ of the harmonic map equations read
\begin{equation} \label{eqn:r}
	\Delta_g r_j = \frac{\partial}{\partial r_j} Q(\nabla \xi,\nabla \xi) \geq 2a.
\end{equation}
Multiplying by $\chi^2 r_j$, where $\chi\in C^\infty_0(M_{\delta_0})$ is a cut-off function, and integrating by parts, we obtain
\[
	\int_{M_{\delta_0}} \chi^2 |\nabla r_j|^2 \leq 4 \left(\sup_{M_{\delta_0}} r_j^2\right) \int_{M_{\delta_0}} |\nabla \chi|^2.
\]
Now going back to Equation\eqref{eqn:r}, multiplying by $\chi^2$ and integrating by parts, we obtain
\[
	\int_{M_{\delta_0}} \chi^2 Q(\nabla\xi,\nabla\xi) \leq \frac1a \int_{M_{\delta_0}} \chi \nabla\chi \nabla r_j
	\leq \frac1a \left( \int_{M_{\delta_0}} |\nabla r_j|^2 \right)^{1/2}
	\left( \int_{M_{\delta_0}} |\nabla \chi|^2 \right)^{1/2} 
\]
Combining these two inequalities gives the required local energy estimate $E_\Omega(\varphi_\delta) \leq C$ where $C$ depends on $\Omega\Subset M_{\delta_0}$ but is independent of $\delta$.

We can now establish a pointwise bound on the derivatives using the Bochner identity:
\[
	\Delta_g |d\varphi_\delta|^2 = |\nabla d\varphi_\delta|^2 + \operatorname{Ric_M}(d\varphi_\delta,d\varphi_\delta)
	- \operatorname{Riem_N}(d\varphi_\delta,d\varphi_\delta,d\varphi_\delta,d\varphi_\delta).
\]
Since $\operatorname{Riem_N}\leq0$, it follows
\[
	\Delta_g |d\varphi_\delta|^2 \geq - C |d\varphi_\delta|^2,
\]
where $C=\sup_{M_{\delta_0}}|\operatorname{Ric_M}|$, and standard elliptic theory immediately gives an $L^\infty$ on $|d\varphi_\delta|^2$ from the $L^1$ bound already obtained. One can now bootstrap using the harmonic map equations~\eqref{harmmap} to get pointwise bounds on any number of derivatives independent of $\delta$.
\end{proof}

\section{Applications in general relativity}

The applications to general relativity of harmonic maps with prescribed singularities stem from the following fact or variations thereof.
Under the hypothesis of \emph{stationarity} and $(D-3)$-\emph{axisymmetry}, i.e.\ admitting an isometric action of the group $\R\times U(1)^{D-3}$, the Einstein vacuum equations reduce to a harmonic map problem
$\varphi\colon \R^3\setminus\Gamma\to N$ where $\Gamma\subset$ $z$-axis, and $N=SL(D-2,\R)/SO(D-2)$. Note that $N$ is a symmetric space of non-compact type and rank $D-3$. We recall that the \emph{rank} is the maximal dimension of a \emph{flat}, i.e.\ a totally geodesic submanifold with vanishing curvature.
For the sake of brevity, we will not expound here further on this reduction; see~\cite{khuriweinsteinyamada2017}*{Section~2} for more details.
We only comment on two essential features: the quotient space is 2 dimensional, and the Einstein equations decouple so that the system becomes semi-linear.

\subsection{Non-degenerate vacuum black holes in 4-d gravity}

In this case, the target is the hyperbolic plane $N=SL(2,\R)/SO(2)=\H^2$~\cite{weinsteinMRL}.
This line of investigation, which begun in the static linear case by Bach-Weyl~\cite{bachweyl}, was pursued further in a series of papers by this author starting with~\cite{weinsteinCPAM1990}, culminating with~\cite{weinsteinMRL}, in which an earlier version of Theorem~\ref{existence-uniqueness} was proved.

The metrics considered are of the form
\begin{equation} \label{metric4d}
 g=-f^{-1}\rho^2 dt^2+f(d\phi+vdt)^2+f^{-1}e^{2\sigma}(d\rho^2+dz^2),
\end{equation} 
and are asymptotically flat.
The \emph{twist potential} is defined by $d\omega=*(\xi\wedge d\xi)$, and is related to the metric function $v$ by $dv=-\rho f^{-2} *_2 d\omega$, where $*_2$ is the Hodge dual in the flat $\rho z$-plane. The one form $\omega$ measures how far $\xi$ is from being hypersurface orthogonal. It vanishes in the \emph{static} case, indicating the absence of rotation, in which case the Einstein equations reduce to one linear equation~\footnote{This is the case studied by Bach and Weyl in 1917.}. The map blows up on the axis $\Gamma$. The singular boundary conditions can be used to specify the number of black holes, their mass and their angular momenta. Once that data is prescribed, one constructs a model map $\varphi_0$, and there is, by Theorem~\ref{existence-uniqueness}, a unique solution to the harmonic map problem.

Also as pointed out in~\cite{weinsteinMRL}, given such a harmonic map, one can reconstruct a solution of the original Einstein vacuum equations. The metric~\eqref{metric4d} thus constructed can be extended across $\Gamma$ if and only if there are no conical singularities, and if the metric components are sufficiently smooth up to the axis. The latter condition, which we call \emph{analytic regularity}, was established in~\cite{weinsteinCPAM1990}; see also generalizations in~\cites{litian92,nguyen2011}. The angle deficiency $e^{\beta}=\lim_{\rho\to0}\rho^2 e^{2\sigma} f^{-2}$ is constant along any connected component of the axis, and it in fact vanishes on the two unbounded component~\cite{khuriweinsteinyamada2018}. With two black holes, it is related to the gravitational force acting between them $F=1/4 (e^{-\beta/2}-1)$. In the static linear case, Bach and Weyl first calculated $F=m_1 m_2/(R^2-(m_1+m_2)^2)$, where $m_1$ and $m_2$ is a measure of the masses of the two bodies, and $R-m_1-m_2$ is a measure of their separation. It is conjectured that it is impossible to balance several rotating black holes. This is equivalent to a conical singularity always being present. Some results were obtained in~\cites{litian,weinsteinTAMS}.

These results were generalized to the Einstein-Maxwell equations, and also later further generalized to the Einstein-Abelian-Yang-Mills equations, in which case $N=\H^2_\C$, and $\H^k_\C$ respectively~\cite{weinsteinMRL}.

\subsection{Degenerate black holes in 4-d gravity}

The generalization of these results to the \emph{degenerate} horizon case, sometimes also called the~\emph{extreme} case, is primarily motivated by an idea of S.~Dain~\cite{dain2008} which allows one to prove lower bounds on the ADM mass of axisymmetric data in terms of angular momentum and charge. Dain showed that the ADM mass is bounded below by a renormalized harmonic map energy and that this renormalized energy is minimized by the harmonic maps corresponding to extreme black holes.
In particular, Dain~\cite{dain2008} proved the inequality $m\geq |J|$ where $m$ is the ADM mass and $J$ the angular momentum. This was generalized to multiple black holes where we proved $m\geq \mathcal F$, and $\mathcal F$ is the harmonic map renormalized energy for the corresponding extreme black hole~\cite{chruscielliweinstein}. Dain's result was later generalized to the charged single black hole case~\cite{chruscielcosta} $m^2\geq\frac12\left(q + \sqrt{q^2+4J^2} \right)$, and we again generalized this to the multiple black hole case~\cite{khuriweinstein2015}. 

A key ingredient in these results is the construction of the minimizer corresponding to multiple black holes. The existence result for harmonic maps with prescribed singularities described in Section~\ref{harmap} is used to achieve this goal.

In the multiple black hole case, both in vacuum and in the charged case, the lower bound is not explicit, but it is conjectured that it is the same as in the single black hole case with $q$ and $J$ denoting the total charge and angular momentum.

\subsection{Black holes in 5-d gravity}

In this latest application of harmonic maps with prescribed singularities, we considered the problem of vacuum black holes in 5-d with non-spherical topology. In 4-d, a result of Hawking~\cite{hawking1972} shows that the cross section of any connected component of the event horizon in an asymptotically flat (AF) stationary spacetime satisfying the dominated energy condition (DEC), has positive Euler characteristic, and hence must be topologically a $2$-sphere. However, the higher dimensional analogous result by Galloway and Schoen~\cite{schoengalloway} only shows that black hole boundaries are $(n-2)$-dimensional Riemannian manifolds with positive Yamabe invariant. In particular, in 5-d this allows prime 3-d manifolds of positive Yamabe type, $S^3$, $S^2\times S^1$, and the lens spaces $L(p,q)$. Some explicit examples with $S^2\times S^1$ horizon topology, referred to as Black Rings, are known~\cites{emparanreall2008,pomeranskysenkov}.

The target for the harmonic map problem is now the symmetric space $N=SL(3,\R)/SO(3)$. 
In~\cite{khuriweinsteinyamada2017}, we constructed harmonic maps corresponding to AF spacetimes, and in~\cite{khuriweinsteinyamada2018} solutions corresponding to asymptotically locally Euclidean (ALE) spacetimes, i.e.\ whose spacelike infinity is modeled on $(a,\infty)\times L(p,q)$, as well as asymptotically Kaluza-Klein (AKK), i.e.\ with spacelike infinity modeled on $(a,\infty)\times S^1\times S^2$. Here too, we are able to construct multiple black hole solutions, provided the necessary topological obstructions are avoided. 

The metrics we consider are now of the form
\begin{equation} \label{metric5d}
g=f^{-1}e^{2\sigma}(d\rho^2+dz^2)-f^{-1}\rho^2 dt^2
+f_{ij}(d\phi^{i}+v^{i}dt)(d\phi^{j}+v^{j}dt),
\end{equation}
where $\rho^2$ is minus the determinant of the metric on the orbits of $\R\times U(1)^2$, $f_{ij}$ is the metric on the orbits of $U(1)^2$ and $f$ is the determinant of $f_{ij}$.
Also as before, the resulting metrics can be extended across the axis if and only if analytic regularity holds, and there are no conical singularities.

The topology of the horizon is prescribed via the \emph{rod structure}, a pair of relatively prime integers $(m,n)$ assigned to each \emph{axis rod} $\subset\Gamma$. Rods are subintervals of $\Gamma$ on which the rank of $f_{ij}$ is either $1$ on axis rods, or $2$ on horizon rods. The end points of horizon rods are \emph{poles}, points separating axis rods are called \emph{corners}. The rod structure $(m,n)$ of an axis rod is defined to be such that $m\partial_{\phi_1}+n\partial_{\phi_2}$ is the generator of $U(1)^2$ which degenerates on that rod. The topology of the horizon is determined by the rod structures of the two adjacent rods. Some simple examples of rod structure configurations are given in Figure~\ref{rods}. It turns out that the potentials $v^i$ are constant on each component of $\Gamma$. The difference between the two constant values of $v^i$ adjacent to any horizon determines the angular momentum associated with that generator for that black hole. The data consists of the rod structure configuration and these constants. Given any data satisfying a mild admissibility condition, we showed in~\cite{khuriweinsteinyamada2017} that a model map could be constructed, and hence a unique harmonic map exists corresponding to this data.

\begin{figure}
	\includegraphics[]{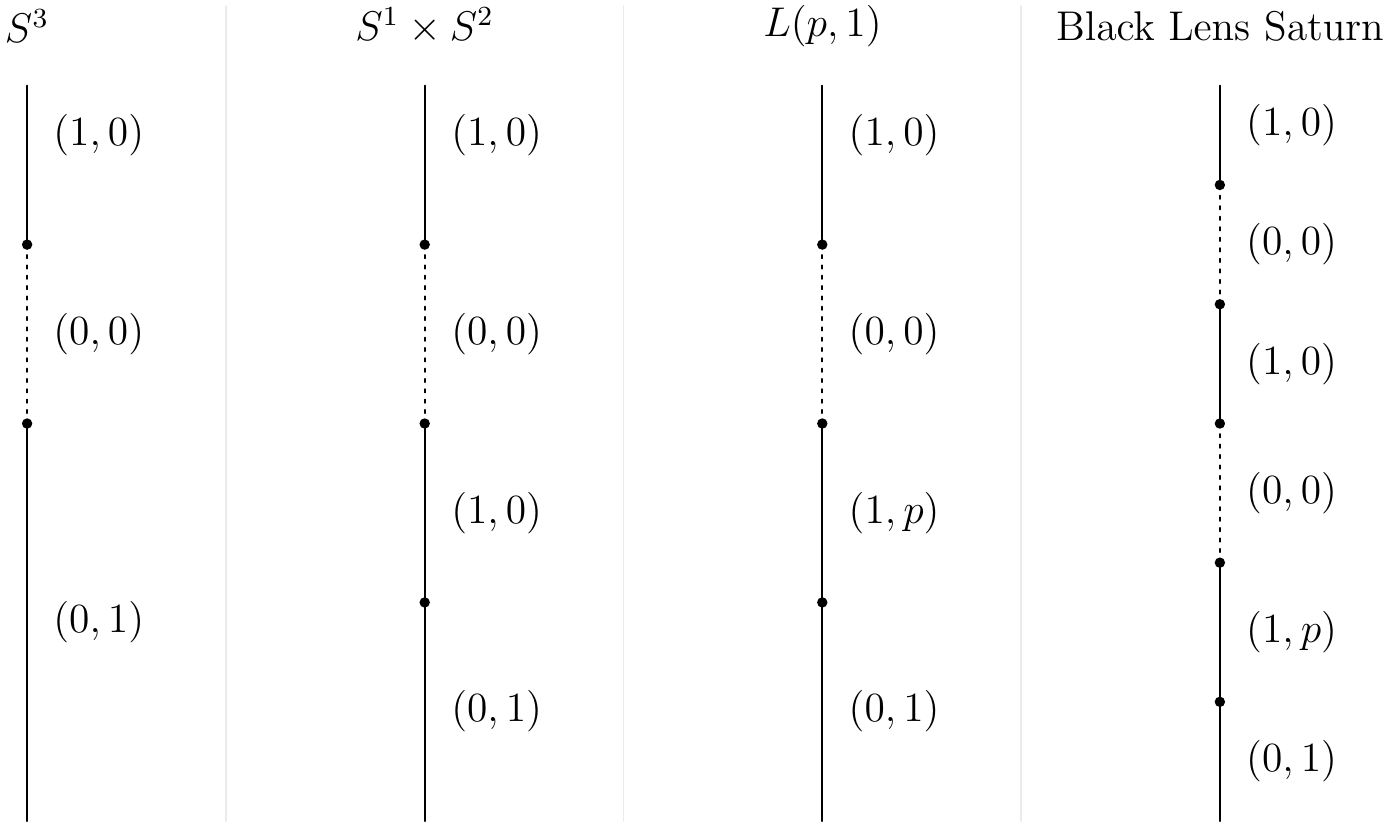}
	\caption{Some examples of rod structure configurations.} \label{rods}
\end{figure}

In~\cite{khurimatsumotoweinsteinyamada}, we analyzed the topology of the domain of outer communications of these solutions using \emph{linear plumbing}.  Plumbing is a construction used to glue together two disk bundles by identifying the bundles over two neighborhoods, crisscrossing base and fiber. Utilizing this technique we were able, given any rod structure configuration, to determine that the topology of the 4$D$ slice in our solutions is a sequence of disk bundle plumbings.

\section{Conclusion}

In this paper, we sketched the proof of an existence and uniqueness result for harmonic maps with prescribed singularities and surveyed a number of applications to general relativity, particularly to black hole existence and uniqueness problems in dimensions 4 and 5, both with non-degenerate and degenerate horizons. In particular, the applications in dimension 5 provided a glimpse into an interesting interface between general relativity and low dimensional topology.

We believe there are many more applications of this versatile method which should yield answers to a number of interesting unanswered questions.


\bibliography{$HOME/Dropbox/.texmf/mybib}

\end{document}